\documentclass[12pt]{article}

\usepackage[left=2.5cm, right=2.5cm, top=2.5cm, bottom=2.5cm]{geometry}
\usepackage{amssymb,amsthm,amsmath}
\usepackage{braket}
\usepackage{xcolor}
\usepackage{algorithm}
\usepackage[noend]{algpseudocode}
\usepackage{authblk}

\usepackage[final=true,colorlinks = true,allcolors = {blue}]{hyperref}

\newtheorem{theorem}{Theorem}
\newtheorem{proposition}{Proposition}
\newtheorem{lemma}{Lemma}
\newtheorem{corollary}{Corollary}

\newcommand{\V}{\mathcal{V}}
\renewcommand{\S}{\mathcal{S}}
\newcommand{\E}{\mathcal{E}}

\newcommand{\tO}{\widetilde{O}}
\newcommand{\Id}{\mathbb{I}}
\newcommand{\Hi}{\mathcal{H}}
\newcommand{\complex}{\mathbb{C}}

\begin{document}

\title{Quantum Walk Sampling by Growing Seed Sets}
\author[ ]{Simon Apers}
\affil[ ]{Inria, France and CWI, the Netherlands}
\affil[ ]{\tt simon.apers@inria.fr}
\setcounter{Maxaffil}{0}
\renewcommand\Affilfont{\itshape\small}

\begin{titlepage}
\clearpage\maketitle
\thispagestyle{empty}

\abstract{
This work describes a new algorithm for creating a superposition over the edge set of a graph, encoding a quantum sample of the random walk stationary distribution.
The algorithm requires a number of quantum walk steps scaling as $\tO(m^{1/3} \delta^{-1/3})$, with $m$ the number of edges and $\delta$ the random walk spectral gap.
This improves on existing strategies by initially growing a classical seed set in the graph, from which a quantum walk is then run.

The algorithm leads to a number of improvements: (i) it provides a new bound on the setup cost of quantum walk search algorithms, (ii) it yields a new algorithm for $st$-connectivity, and (iii) it allows to create a superposition over the isomorphisms of an $n$-node graph in time $\widetilde{O}(2^{n/3})$, surpassing the $\Omega(2^{n/2})$ barrier set by index erasure.}

\end{titlepage}

\section{Introduction and Summary}

Sampling from the stationary distribution of a random walk is a common and valuable tool in the design of algorithms \cite{sinclair2012algorithms}.
It underlies the Markov chain Monte Carlo paradigm, and plays a central role in a wide range of approximation algorithms for graph problems.
In this work we investigate the quantum counterpart of this task - generating quantum samples from the random walk stationary distribution.
Given query access to some graph $G = (\V,\E)$ with $m$ edges, we wish to create the quantum state
\begin{equation} \label{eq:qsample}
\ket{\pi}
= \frac{1}{\sqrt{m}} \sum_{(i,j)\in\E} \ket{i,j},
\end{equation}
which is a superposition over the edges of the graph.
Measuring the first register of this state, and discarding the second register, indeed returns the random walk stationary distribution.
Creating such a quantum sample of a classical stationary distribution forms a crucial primitive for a range of algorithms: the so-called ``setup cost'' in quantum walk search algorithms \cite{magniez2011search,krovi2016quantum} refers to the cost of generating a state such as $\ket{\pi}$, quantum algorithms for speeding up MCMC \cite{aharonov2003adiabatic,somma2008quantum,wocjan2008speedup,poulin2009sampling} build on the possibility of efficiently creating quantum samples, and a number of quantum algorithms for solving graph problems \cite{watrous2001quantum,jarret2018quantum} require the generation of a superposition over the edges of a graph.

We develop a new quantum algorithm for creating the quantum sample \eqref{eq:qsample}, given only local query access to the graph.
Our algorithm improves the query and time complexity of the folklore approach to quantum sampling from $\tO(m^{1/2}\delta^{-1/2})$ to $\tO(m^{1/3}\delta^{-1/3})$.
We do so by growing a classical seed set from the initial node.
This incurs a payoff in the space complexity, increasing it from $\tO(1)$ to $\tO(m^{1/3} \delta^{-1/3})$.
As a demonstration of our algorithm, we discuss a new approach to solving $st$-connectivity: generate a superposition over the connected components of $s$ and $t$, and compare these states.
This approach generalizes the notorious quantum state generation strategy for solving graph isomorphism.
Concerning the latter, we show that our algorithm allows to create a superposition over the isomorphisms of a given $n$-vertex input graph in $\tO(2^{n/3})$ steps.
This surpasses the $\Omega(2^{n/2})$ index erasure barrier by Ambainis et al \cite{ambainis2011symmetry}.
In a similar way we can create a superposition over the elements of a black box group in $\tO(2^{n/3})$ steps, where $2^n$ is the number of group elements.

\vspace{5mm}\noindent
\textit{\textbf{Query Model.}}\hspace{5mm}
We assume throughout this work that we only have ``local'' query access to some graph $G = (\V,\E)$: we are give an initial node $j \in \V$, and we can query for its degree and neighbors.
Such queries fall under the so-called \textit{adjacency array model} \cite{durr2006quantum} or \textit{bounded degree model} \cite{goldreich1997property} (although we do not assume the degree is bounded), which is very natural when studying random walk algorithms.
However, departing from these models, and justifying the term ``local'', we will not assume direct access to or prior knowledge about $\V$, apart from the initial node.
For comparison, in \cite{durr2006quantum} the node set $\V$ is given as a list, and in \cite{goldreich1997property} access to uniformly random nodes is assumed.
In this sense our work is in line with graph exploring algorithms as considered e.g.~in \cite{spielman2013local}, or more recently in \cite{chiericetti2016sampling}.

Since our algorithm strongly builds on the use of quantum walks, we will alternatively express the complexity of our results as a function of the number of quantum walk steps.
Also in such case the denominator ``local'' query access is justified, since a single quantum walk step from a certain node only accesses the neighbors of that node.

\vspace{5mm}\noindent
\textit{\textbf{Quantum Walk Sampling Algorithm.}}\hspace{5mm}
Our algorithm builds on the folklore approach to creating $\ket{\pi}$, discussed in e.g.~\cite{richter2007quantum,wocjan2008speedup,poulin2009sampling,orsucci2015faster}.
Starting from some initial state $\ket{j}$ localized on a node $j \in \V$, this approach combines quantum phase estimation and amplitude amplification on the quantum walk operator associated to the graph.
We detail this scheme in Section \ref{sec:walks}.
The scheme requires $\tO(m^{1/2}\delta^{-1/2})$ QW steps on the graph, where $\delta$ is the random walk spectral gap, and the factor $m^{1/2}$ stems from the small projection of the initial state onto $\ket{\pi}$.

In the present work we improve on this scheme by initially doing some ``classical work'': we first use classical means to grow a seed set around the initial vertex.
Briefly ignoring the $\delta$-dependency, we grow the set to have size $\Theta(m^{1/3})$.
We can then use a QRAM data structure to generate and reflect around a quantum superposition over this set, which now has a $\Omega(m^{-1/3})$ overlap with the target state.
Reinvoking the folklore scheme from this state then allows to retrieve $\ket{\pi}$, now only requiring $\tO(m^{1/3})$ queries.
This approach leads to the following result.
\begin{theorem}
Given a lower bound $\gamma \leq \delta$ on the spectral gap, it is possible to create the quantum state $\ket{\pi}$ using $\tO(m^{1/3}\gamma^{-1/3})$ time, space and QW steps.
\end{theorem}
Apart from the log-factors, the combined dependency on $m$ and $\delta$ is optimal.
Indeed it is tight on e.g.~the cycle graph, which has $m = n$ and $\delta = n^{-2}$, giving an $\tO(n)$ steps algorithm.
Since the diameter of the cycle is $\Omega(n)$, this is optimal when assuming local query access.
We also note that, if in addition we are given a bound $D \geq d_M$ on the maximum degree (in e.g.~the array model this is always given), then we can implement our algorithm using $\tO(m^{1/3}\gamma^{-1/3}D^{1/3})$ degree and neighbor queries.

The algorithm gives a direct bound on the so-called setup cost of quantum walk search algorithms in the MNRS framework \cite{magniez2011search} as a function of the update cost (i.e., the cost of implementing a quantum walk step).
The increased space complexity of our algorithm, $\tO(m^{1/3}\delta^{-1/3})$ as compared to $\tO(1)$ for the folklore approach, is very similar to the payoff in space versus time or query complexity in the collision finding algorithm of Brassard et al \cite{brassard1997quantum} and the element distinctness algorithm of Ambainis \cite{ambainis2007quantum}.

\vspace{5mm}\noindent
\textit{\textbf{Application to $st$-connectivity.}}\hspace{5mm}
Our QW sampling algorithm yields a new approach for solving $st$-connectivity, somewhat similar to the approach taken by Watrous in \cite{watrous2001quantum}: generate a superposition over the edges in the connected components of $s$ resp.~$t$, and compare the resulting states.
As we prove in Proposition \ref{prop:st}, this requires $\tO(m^{1/3} \gamma^{-1/3})$ QW steps, where $\gamma$ is a lower bound on the spectral gaps of the connected components of $s$ and $t$.
Our algorithm outperforms the existing quantum algorithms for $st$-connectivity \cite{durr2006quantum,belovs2012span,belovs2013quantum,jarret2018quantum} on for instance sparse graphs with a good spectral gap.

The approach generalizes a well-known strategy to solving graph isomorphism on a quantum computer \cite{aharonov2003adiabatic} (called ``component mixing'' in \cite{lutomirski2011component}): generate superpositions over the isomorphisms of each graph, and compare the resulting states.
In \cite{ambainis2011symmetry}, Ambainis et al aimed to prove a lower bound on this approach by abstracting it to the so-called index erasure problem.
For this generalized problem, they prove a lower bound of $\Omega(2^{n/2})$.
They argue that the same bound holds for creating a superposition over graph isomorphisms, be it under the condition that the algorithm makes no use of the structure of the problem.
We show that, by exploiting the structure of the problem, we can indeed use our quantum walk sampling algorithm to surpass this bound.
Thereto we consider the graph whose node set consists of isomorphisms of the input graph, and whose edge set arises from performing pairwise transpositions on the nodes (i.e., on the adjacency matrices of the isomorphisms).
Using our quantum walk sampling algorithm on this graph yields the following corollary.
\begin{corollary}
Given an $n$-node input graph $g$, it is possible to create a superposition over the isomorphisms of $g$ in $\tO(2^{n/3})$ steps.
\end{corollary}
\noindent
Completing the associated $st$-connectivity algorithm, we find an $\tO(2^{n/3})$ quantum algorithm for graph isomorphism.
Using the existing quantum algorithms for $st$-connectivity, this approach would require $\Omega(2^{n/2})$ steps.
Clearly the improved performance still falls terribly short of current (classical) algorithms for graph isomorphism, most notably the quasi-polynomial algorithm by Babai \cite{babai2016graph}, yet it provides a clear demonstration of how the readily accessible structure of the problem allows to surpass the index erasure bound.

A similar strategy exists for solving the group non-membership problem on a quantum computer, as proposed by Watrous \cite{watrous2000succinct}, requiring to generate a superposition over the elements of a finite black box group.
Using the random walk algorithm by Babai \cite{babai1991local} for generating uniformly random group elements, we can similarly generate this superposition in $\tO(2^{n/3})$ steps, when $2^n$ is the number of group elements.

\vspace{5mm}\noindent
\textit{\textbf{Open Questions.}}\hspace{5mm}
This work leaves open a number of questions and possible applications, some of which we summarize below:

\begin{itemize}
\item
\textit{Quantum sampling for general Markov chains or stoquastic Hamiltonians.}
In this work we only consider the quantum sampling problem for random walks.
Generalizing our approach to more general Markov chains could lead to improvements on quantum MCMC algorithms \cite{aharonov2003adiabatic,somma2008quantum}, or the preparation of many body ground states \cite{poulin2009sampling} and Gibbs states \cite{van2017quantum}.
The main bottleneck to such generalization seems to be the classical construction of seed sets which have an appropriate overlap with the goal quantum state.
Even more generally, one could consider the preparation of ground states of arbitrary Hamiltonians.
For e.g.~the special case of stoquastic Hamiltonians, which are known to have a nonnegative ground state, it should be possible to construct a seed set with improved overlap with the ground state.

\item
\textit{Faster quantum fast-forwarding.}
In \cite{apers2018quantum} a quantum algorithm was proposed for quantum sampling a $t$-step Markov chain.
If the Markov chain has transition matrix $P$, and is started from a node $i$, the algorithm has complexity $\tO(\|P^t\ket{i}\|^{-1}\, t^{1/2}) \in \tO(m^{1/2}\, t^{1/2})$.
Using ideas from the present work, it seems very feasible that we can improve this complexity to $\tO(\|P^t\ket{i}\|^{-2/3}\, t^{1/2}) \in \tO(m^{1/3}\, t^{1/2})$.
Rather than using a breadth-first search to grow the seed set, as in the present work, it seems more suitable to use random walk techniques as in \cite{spielman2013local,andersen2009finding}.
As a byproduct, this would yield an improved quantum expansion tester, combining the speedups of \cite{ambainis2011quantum} and \cite{apers2018quantum}.

\item
\textit{Quantum search in $\sqrt{HT}$.}
Our algorithm does not suffer from the so-called ``symmetry barrier'' in quantum algorithms: we can go from $\ket{j}$ to $\ket{\pi}$ more easily then from $\ket{\pi}$ to $\ket{j}$.
Indeed, if for instance the underlying graph is an expander, then the former takes $O(n^{1/3})$ queries, whereas the latter takes $\Omega(n^{1/2})$ queries by the search lower bound.

An open problem related to this is the following: given an initial node $s$ in a graph, can we find a node $t$ in $O(HT_{s,t}^{1/2})$ QW steps, with $HT_{s,t}$ the hitting time from $s$ to $t$?
Currently the best algorithm for this problem is by Belovs \cite{belovs2013quantum}, which solves it in $O(CT_{s,t}^{1/2})$, with $CT_{s,t} = H_{s,t} + H_{t,s}$ the commute time between $s$ and $t$.
Since the commute time is symmetric between $s$ and $t$, this obeys the aforementioned symmetry barrier.
However, the commute time can be much larger than the hitting time from $s$ to $t$, hence the open question of whether we can improve this performance to $O(HT_{s,t}^{1/2})$. thereby necessarily breaking this symmetry e.g.~by using our techniques.
\end{itemize}

\vspace{5mm}\noindent
\textit{\textbf{Outline.}}\hspace{5mm}
In Section \ref{sec:prelim} we discuss the graph and query model (Section \ref{sec:query}), and provide the necessarily preliminaries on random walks and quantum walks (Section \ref{sec:walks}).
In Section \ref{sec:QW-sampling} we propose an algorithm for growing a classical seed set (Section \ref{sec:BFS}), we discuss the QRAM data structure (Section \ref{sec:QRAM}), and we propose our QW sampling algorithm (Section \ref{sec:QW-alg}).
Finally in Section \ref{sec:app-st} we discuss the application of our QW sampling algorithm for solving $st$-connectivity (Section \ref{sec:general-st}), and we demonstrate it for the special case of graph isomorphism testing (Section \ref{sec:graph-iso}).

\section{Preliminaries: Queries and Walks} \label{sec:prelim}

\subsection{Graph and Query Model} \label{sec:query}
Throughout the paper we assume local query access to an undirected graph $G = (\V,\E)$, with $\E$ a subset of the ordered pairs $\V \times \V$, such that $(i,j) \in \E \Leftrightarrow (j,i) \in \E$.
We denote $|\V| = n$ and $|\E| = m$.
For any $\S \subseteq \V$, we let $\E(\S)$ denote the set of edges starting in $\S$, i.e.,
\[
\E(\S) = \{(i,j)\in\E \mid i \in \S\}.
\]
For any $i \in \V$, we let $d(i) = |\E(\{i\})|$ denote the degree of $i$, the maximum degree $d_M = \max_{i\in\V} d(i)$, and $d(\S) = |\E(\S)| = \sum_{i\in\S} d(i)$ denotes the total degree of a set $\S \subseteq \V$.
A single query consists of either of the following:
\begin{itemize}
\item
\textit{degree query}: given $i \in \V$, return degree $d(i)$
\item
\textit{neighbor query:} given $i \in \V$, $k \in [d(i)]$, return $k$-th neighbor of $i$
\end{itemize}

As an alternative query model we will also consider the quantum walk model, or so-called MNRS framework, as proposed in \cite{magniez2011search} in the context of quantum walk search.
The model associates abstract costs to different operations\footnote{They actually consider a more general model, associated to a reversible Markov chain over $G$. We consider the special case where the Markov chain is a random walk.}:
\begin{itemize}
\item
\textit{setup cost:}
the cost of preparing the quantum sample $\ket{\pi} = m^{-1/2} \sum_{(i,j)\in\E} \ket{i,j}$
\item
\textit{update cost:}
the cost of implementing a quantum walk step.
See Section \ref{sec:walks} for details.
\end{itemize}
For search problems an additional \textit{checking cost} is considered, yet this will not be relevant here.
In \cite{cade2016time} it is proven that the update cost or quantum walk step for a node $i$ can be simulated using $O(d(i)^{1/2})$ degree and neighbor queries.
From our work it follows that the setup cost can be simulated using $\tO(m^{1/3} \delta^{-1/3})$ QW steps, or $\tO(m^{1/3} d_M^{1/3} \delta^{-1/3})$ degree and neighbor queries.

\subsection{Random Walks and Quantum Walks} \label{sec:walks}
From some initial seed vertex $j \in \V$, we can use degree and neighbor queries to implement a random walk over $\V$.
The transition matrix $P$ describing such a walk is defined by $P(i,j) = 1/d(i)$ if $(i,j) \in \E$, and $P(i,j) = 0$ elsewhere.
If the graph is connected and nonbipartite, then the random walk converges to its stationary distribution $\pi$, defined by $\pi(i) = d(i)/m$ for any $i \in \V$.
If we order the eigenvalues of $P$ (with multiplicities) as $1 = \lambda_1 \geq \lambda_2 \geq \dots \geq \lambda_n \geq -1$, then the rate at which the walk converges to $\pi$ is bounded by the spectral gap $\delta = 1 - \max\{|\lambda_2|,|\lambda_n|\}$ \cite{levin2017markov}.

Quantum walks (QWs) form an elegant and nontrivial quantum counterpart to random walks on graphs.
Following the exposition in \cite{magniez2011search}, they are naturally defined over a vector space associated to the edge set
\[
\Hi_\E
= \mathrm{span}_\complex\big\{ \ket{i,j} \mid i,j \in \V \big\}.
\]
A quantum walk over $\Hi_\E$ is now defined as the unitary operator $W = S R_\E$, where $R_\E$ is a reflection around the subspace $\mathrm{span}_\complex\{\ket{\psi_i} \mid i \in \V \}$, with
\[
\ket{\psi_i}
= \frac{1}{\sqrt{d(i)}} \sum_{(i,j)\in\E} \ket{i,j},
\]
and $S$ represents the swap operator defined by $S\ket{i,j} = \ket{j,i}$.
The cost of implementing the QW operator $W$ is called the update cost, but can alternatively be implemented using $O(d_M^{1/2})$ degree and neighbor queries, and $\tO(1)$ elementary operations.

The spectrum of $W$ is carefully tied to the spectrum of the original random walk matrix $P$, as was seminally proven by Szegedy in \cite{szegedy2004quantum} and Magniez et al in \cite{magniez2011search}.
For the purpose of this work, we abstract the following lemma.
We say that $W$ has a phase gap $\Delta$ if for every eigenvalue $e^{i\theta} \neq 1$ of $W$ it holds that $|\theta| \geq \Delta$.
We also recall the state $\ket{\pi} = m^{-1/2} \sum_{(i,j)\in\E} \ket{i,j}$.
\begin{lemma}[\cite{szegedy2004quantum,magniez2011search}] \label{lem:QW-gap}
Let $P$ be the random walk transition matrix having spectral gap $\delta$.
Then the quantum walk operator $W$ has a phase gap $\Delta \in \Theta(\sqrt{\delta})$, and $\ket{\pi}$ is the unique eigenvalue-1 eigenvector in the subspace $\mathrm{span}_\complex\{\ket{\psi_i} \mid i \in \V\}$.
\end{lemma}

From this lemma, combined with the quantum algorithms for phase estimation and amplitude amplification, we can derive the folklore approach to quantum walk sampling, discussed in for instance \cite{richter2007quantum,wocjan2008speedup,poulin2009sampling,orsucci2015faster}.
Since we will use it as a subroutine, we summarize it below.
For a general subset $\S \subseteq \V$, we denote the state $\ket{\S} = d(\S)^{-1/2} \sum_{(i,j)\in\E(\S)} \ket{i,j}$.
\begin{proposition} \label{prop:folk}
Given an initial set $\S \subseteq \V$ and a lower bound $\gamma \leq \delta$, we can generate a state $\epsilon$-close to $\ket{\pi}$ using an expected number of $O(d(\S)^{-1/2} m^{1/2} \gamma^{-1/2} \log \epsilon^{-1})$ calls to $W$, $O(d(\S)^{-1/2} m^{1/2})$ reflections around $\ket{\S}$, and an additional $O(\log\epsilon^{-1} \log^2\gamma^{-1})$ time and space complexity.
\end{proposition}
\begin{proof}
Let the operator $U$ be defined by the amplified quantum phase estimation algorithm, as used in \cite[Theorem 6]{magniez2011search}.
For some integer $k$, this operator maps an initial state $\ket{\S}$ to the state
\[
U\ket{\S}\ket{0}
= \braket{\pi|\S} \ket{\pi}\ket{0} + \ket{\Gamma},
\]
where $\ket{\Gamma}$ is such that $\|(\Id \otimes \ket{0}\bra{0}) \ket{\Gamma}\| \leq 2^{-k}$.
The operator $U$ can be implemented using $O(k \Delta^{-1}) \in O(k \gamma^{-1/2})$ calls to $W$ and $W^\dag$, and $O(k\log^2\gamma^{-1})$ additional space and elementary gates.

On this state we can invoke the amplitude amplification scheme from \cite[Theorem 3]{brassard2002quantum} to retrieve the projection of $U\ket{\S}\ket{0}$ on the image of $\Id \otimes \ket{0}\bra{0}$, which is $2^{-k}$-close to $\ket{\pi}$.
This requires an expected number of $\Theta(|\braket{\S|\pi}|^{-1})$ calls to $U$, $U^\dag$ and the reflection operator $\Id \otimes (2\ket{0}\bra{0} - \Id)$.
We prove the proposition by choosing $k \in \Theta(\log\epsilon^{-1})$ and noting that $|\braket{\S|\pi}| = d(\S)^{1/2} m^{-1/2}$.
\end{proof}
\noindent
On a general graph, and starting from some initial node $\S = \{i\}$, this scheme requires $\tO(d(i)^{-1/2} m^{1/2} \gamma^{-1/2}) \in \tO(m^{1/2} \gamma^{-1/2})$ QW steps, or $\tO(m^{1/2} d_M^{1/2} \gamma^{-1/2})$ degree and neighbor queries.

\section{Quantum Walk Sampling} \label{sec:QW-sampling}
In this section we elaborate our scheme for quantum walk sampling.
We separately address the process for growing a seed set, the QRAM data structure that we require, and their combination with the folklore QW sampling routine.

\subsection{Growing a Seed Set} \label{sec:BFS}
We propose the following algorithm to grow a seed set in the graph.
It is a variation on the breadth-first search algorithm, returning an edge set of given size.

\vspace{3mm} {\centering
\begin{minipage}{.9\linewidth}
\begin{algorithm}[H]
\caption{Breadth-First Edge Search} \label{alg:bfs}
\normalsize
\textbf{Input:} initial node $i$ and query access to a connected graph $G$, integer $M$ \\
\textbf{Do:}
\begin{algorithmic}[1]
\State
create lists $S = \emptyset$ and $E = \emptyset$, and queue $B = (i)$
\While {$B \neq \emptyset$}
  \State
  $i \leftarrow \mathtt{dequeue}(B)$, $\mathtt{add}(S \leftarrow i)$
  \ForAll{$j$ s.t.~$(i,j) \in \E$}
    \If{$j \notin S$}
      \State $\mathtt{add}(E \leftarrow (i,j))$
      \State \algorithmicif\ $|E| = M$ \algorithmicthen\ terminate and output $E$
    \EndIf
    \State \algorithmicif\ $j \notin B$ \algorithmicthen\ $\mathtt{enqueue}(B \leftarrow j)$
  \EndFor
\EndWhile
\end{algorithmic}
\end{algorithm}
\end{minipage}
\par
} \vspace{5mm}

\begin{lemma} \label{lem:BFS}
If $M \leq m$, then Algorithm \ref{alg:bfs} outputs a subset $E \subseteq \E$ with $|E| = M$.
Its time and space complexity, and degree and neighbor query complexity, are $\tO(M)$.
\end{lemma}
\begin{proof}
Assuming the lists are ordered, any of the list and queue operations (enqueueing, dequeueing, adding an element, outputting the size of a list, searching an element in a list) takes polylogarithmic time.
As a consequence, the time complexity will be determined up to log-factors by the number of for-loops before the algorithm terminates.

In every for-loop an edge is considered.
Every edge is encountered at most twice, the first time of which it is added to $E$.
Since the algorithm terminates when $|E|=M$, this implies that the algorithm terminates after less than $2M$ for-loops.
\end{proof}

\noindent
Alternatively we can output the node set $\S \subseteq \V$.
Since $E \subseteq \E(\S)$, we have that $d(\S) \geq M$.

\subsection{Kerenidis-Prakash QRAM} \label{sec:QRAM}
After growing the seed set $\S \subseteq \V$, we wish to use it as a resource for our QW sampling algorithm.
Specifically we will require the generation of and reflection around the superposition $\ket{\S}$ over edges starting in $\S$.
By naive query access to the database containing $\S$, this requires a time complexity $\Omega(d(\S)^{1/2})$ per generation or reflection, which follows from the bound on index erasure \cite{ambainis2011symmetry}.
Since our QW sampling algorithm will require $\Omega(m^{1/3})$ such operations, the total time complexity for $d(\S) \in \Theta(m^{1/3})$ would become $\Omega(m^{1/2})$, thus providing no speedup on the time complexity as compared to the folklore approach.
To remedy this, we use a more efficient QRAM data structure proposed by Kerenidis and Prakash \cite{kerenidis2016quantum} in their quantum recommendation algorithm.
We extract the following result, abstracted from their Theorem 15 (by setting $\mathtt{m}=1$, $\mathtt{n}=n^2$ and inputting entries $(1,(i,j),1)$ for all $(i,j)\in\S$).
\begin{theorem}[Kerenidis-Prakash \cite{kerenidis2016quantum}] \label{thm:QRAM}
Assume we have query access to a set $\S \subseteq \V$.
There exists a classical data structure to store the set $\S$ with the following properties:
\begin{itemize}
\item
the size of the structure is $O(|\S| \log^2(m))$,
\item
the time and query complexity to fill the structure is $O(|\S| \log^2(m))$,
\item
having quantum access to the data structure we can perform the mapping $U: \ket{0} \to \ket{\S}$ and its inverse $U^\dag$ in time $\mathrm{polylog}(m)$.
\end{itemize}
\end{theorem}
\noindent
This easily implies the ability to reflect around $\ket{\S}$ in time $\mathrm{polylog}(m)$: we can rewrite the reflection $2\ket{\S}\bra{\S} - \Id = U (2\ket{0}\bra{0} - \Id) U^\dag$, so that it comes down to implementing $U$, $U^\dag$ and a reflection around the basis state $\ket{0}$.

\subsection{QW Sampling Algorithm} \label{sec:QW-alg}

Building on the seed set and QRAM, we can now propose our quantum sampling algorithm for creating the state $\ket{\pi}$ in $\tO(m^{1/3} \delta^{-1/3})$ time, space and quantum walk steps.

\vspace{3mm} {\centering
\begin{minipage}{.9\linewidth}
\begin{algorithm}[H]
\caption{Quantum Walk Sampling} \label{alg:qsampling}
\normalsize
\textbf{Input:} parameters $\gamma$ and $\epsilon$; initial node $i$ and query access to a graph $G$ \\
\textbf{Do:}
\begin{algorithmic}[1]
\For{$M = 1,2,4,.\,.\,,2^k,\dots$}
\State
use BFS to grow a seed set $\S$ with $d(\S) \in \Theta(M^{1/3}\gamma^{-1/3})$
\State
load $\S$ in QRAM
\State
apply the routine from Proposition \ref{prop:folk} on $\ket{\S}$ for $\widetilde{\Theta}(M^{1/3}\gamma^{-1/3}\log\epsilon^{-1})$ steps
\State
if the routine finished, terminate and return its output
\EndFor
\end{algorithmic}
\end{algorithm}
\end{minipage}
\par
} \vspace{5mm}

\begin{theorem}[Quantum Walk Sampling] \label{thm:qsampling}
If we choose $\gamma \leq \delta$ then Algorithm \ref{alg:qsampling} returns a state $\epsilon$-close to $\ket{\pi}$.
The algorithm requires expected space, time and quantum walk steps in
\[
\tO(m^{1/3} \gamma^{-1/3} \log \epsilon^{-1}).
\]
\end{theorem}
\begin{proof}
The correctness of the algorithm follows from Proposition \ref{prop:folk}.
By this proposition we know that if $\gamma \leq \delta$ and the algorithm terminates, and hence the routine from Proposition \ref{prop:folk} finished, then it effectively outputs a state $\epsilon$-close to $\ket{\pi}$.
The complexity of the algorithm for a fixed $M$ is also easily bounded:
the complexity of steps 2 and 3 is both $\tO(M^{1/3}\gamma^{-1/3})$, which follows from Lemma \ref{lem:BFS} resp.~Theorem \ref{thm:QRAM}.
Step 4 is automatically terminated after $\tO(M^{1/3}\gamma^{-1/3}\log\epsilon^{-1})$ steps, which by Proposition \ref{prop:folk} directly bounds the number of calls to $W$ and reflections around $\ket{\S}$.
By Theorem \ref{thm:QRAM} the complexity of implementing a single reflection around $\ket{\S}$ is $\tO(1)$.
The total complexity for a fixed $M$ is therefore $\tO(M^{1/3}\gamma^{-1/3}\log\epsilon^{-1})$.

What remains to bound is the $M$-value at which the algorithm terminates.
From Proposition \ref{prop:folk} we know that if the number of steps $M^{1/3}\gamma^{-1/3}\log\epsilon^{-1}$ is sufficiently large, i.e.,
\begin{equation} \label{eq:cond-stop}
M^{1/3}\gamma^{-1/3}\log\epsilon^{-1}
\in \Omega(|\braket{\pi|\S}|^{-1} \gamma^{-1/2} \log\epsilon^{-1}),
\end{equation}
then the routine finishes with probability $\Omega(1)$.
From the fact that $\ket{\pi} = m^{-1/2} \sum_{(i,j)\in\E} \ket{i,j}$ and $d(\S) \in \Theta(M^{1/3}\gamma^{-1/3})$ it holds that $|\braket{\pi|\S}| \in \Theta(M^{1/6} \gamma^{-1/6} m^{-1/2})$.
As a consequence, if $M \geq m$ then $|\braket{\pi|\S}| \in \Theta(m^{-1/3} \gamma^{-1/6})$ and hence \eqref{eq:cond-stop} will hold, such that the routine will finish with probability $\Omega(1)$.
The expected number of for-loops is therefore $\log m + O(1)$, with the total complexity scaling as
\[
\tO\bigg( \gamma^{-1/3} \log\epsilon^{-1} \sum_{k=0}^{\log m + O(1)} 2^{k/3} \bigg)
\in \tO(m^{1/3} \gamma^{-1/3} \log\epsilon^{-1}). \qedhere
\]
\end{proof}

Alternatively we are interested in bounding the algorithm in terms of classical queries.
We can naively substitute every quantum walk step for $\tO(\sqrt{d_M})$ degree and neighbor queries, yielding a complexity $\tO(m^{1/3} d_M^{1/2} \gamma^{-1/3})$.
However, if we are given an upper bound $D \geq d_M$, we can improve this complexity by slightly increasing the size of the seed set.
We note that in the array model \cite{durr2006quantum} the degrees are assumed to be known beforehand, so we exactly know $d_M$.
\begin{corollary}
Given an initial node $i$, a lower bound $\gamma \leq \delta$ and an upper bound $D \geq d_M$, we can generate a state $\epsilon$-close to $\ket{\pi}$ in expected space, time, and degree and neighbor queries in
\[
\tO(m^{1/3} D^{1/3} \gamma^{-1/3} \log\epsilon^{-1}).
\]
\end{corollary}
\begin{proof}
We adapt Algorithm \ref{alg:qsampling} by slightly increasing the size of the seed set in step 2 to $\Theta(M^{1/3} D^{1/3} \gamma^{-1/2})$ and decreasing the number of steps in step 4 to $\widetilde{\Theta}(M^{1/3} D^{-1/6} \gamma^{-1/3} \log\epsilon^{-1})$.
Following the proof of Theorem \ref{thm:qsampling}, the algorithm then terminates after $\tO(m^{1/3} D^{1/3} \gamma^{-1/3} \log\epsilon^{-1})$ classical steps and queries, and $\tO(m^{1/3} D^{-1/6} \gamma^{-1/3})$ QW steps.
Now we can substitute each QW step with $\tO(\sqrt{d_M})$ degree and neighbor queries, yielding the claimed complexity.
\end{proof}

\section{Application: st-Connectivity} \label{sec:app-st}

\subsection{General Algorithm} \label{sec:general-st}

Let $\delta^{(s)}$ and $\delta^{(t)}$ denote the spectral gaps of the connected components of $s$ resp.~$t$.
\begin{proposition} \label{prop:st}
Given $s,t \in \V$ and a lower bound $\gamma \leq \delta^{(s)},\delta^{(t)}$, we can decide $st$-connectivity with probability $1-\epsilon$ in $\tO(m^{1/3} \gamma^{-1/3} \log\epsilon^{-1})$ QW steps.
If we are also given an upper bound $D \geq d_M$, then we can do so in $\tO(m^{1/3} D^{1/3} \gamma^{-1/3} \log\epsilon^{-1})$ degree and neighbor queries.
\end{proposition}
\begin{proof}
Given $\gamma$ we can create an $(\epsilon'=1/4)$-approximation $\ket{\psi_s}$ (resp.~$\ket{\psi_t}$) of the superposition $\ket{\pi^{(s)}}$ (resp.~$\ket{\pi^{(t)}}$) over the edges of the connected component of $s$ (resp.~$t$) in $\tO(m^{1/3} \gamma^{-1/3})$ QW steps.
If we also have $D$, then we can do so in $\tO(m^{1/3} d_M^{1/3} \gamma^{-1/3})$ degree and neighbor queries.

If $s$ and $t$ are connected, then $|\braket{\psi_s|\psi_t}| \geq 1 - \epsilon'^2$, whereas if they are not, then $|\braket{\psi_s|\psi_t}| \leq 2\epsilon'$.
We can distinguish these cases by performing the SWAP-test \cite{aharonov2003adiabatic} between these states, using a single copy of both states, and $O(1)$ additional gates.
If $s$ and $t$ are connected, then the test returns 1 with probability $(1-|\braket{\psi_s|\psi_t}|)/2 \leq \epsilon'^2/2 = 1/32$, if $s$ and $t$ are not connected, the test returns 1 with probability $(1-|\braket{\psi_s|\psi_t}|)/2 \geq 1/2-\epsilon' = 1/4$.
Repeating this scheme $O(\log\epsilon^{-1})$ times then allows to decide $st$-connectivity with probability $1-\epsilon$.
\end{proof}

This approach best compares to the following classical scheme: use $\widetilde{\Theta}(n^{1/2})$ independent random walks of length $\Theta(\gamma^{-1})$ from $s$ and $t$ to gather samples from the stationary distributions on the connected components of $s$ resp.~$t$.
If $s$ and $t$ are connected then with constant probability the sample sets will overlap, which follows from the birthday paradox.
This scheme requires $\tO(n^{1/2} \gamma^{-1})$ random walk steps, or equivalently, neighbor queries.
It lies at the basis of the graph expansion tester by Goldreich and Ron \cite{goldreich2011testing}, and the subsequent work on testing closeness of distributions \cite{batu2013testing} and clusterability of graphs \cite{czumaj2015testing}.

In Figure \ref{fig:st} we compare the query complexity of our approach to the existing quantum algorithms for $st$-connectivity.
If no promise is given on negative instances (such as in \cite{jarret2018quantum} in the form of a capacitance $C_{s,t}$), then all former algorithms require $\Omega(n^{1/2})$ queries when maximized over all $(s,t)$-pairs of the graph.
As a consequence, for the graph isomorphism problem treated in the next section, they all have a $\Omega(2^{n/2})$ complexity.
Our approach however has a $\tO(2^{n/3})$ complexity.

\begin{figure}[ht]
\centering
\def\arraystretch{1.5}
\begin{tabular}{|l|l|l|}
\hline
 & query complexity & model \\
\hline
D\"urr et al \cite{durr2006quantum} & $\Theta(n)$ & array \\
\hline
D\"urr et al \cite{durr2006quantum} & $\Theta(n^{3/2})$ & adjacency \\
\hline
Belovs-Reichardt \cite{belovs2012span} & $O(m^{1/2} \, d_{s,t}^{1/2})$ & adjacency \\
\hline
Belovs \cite{belovs2013quantum} & $O(m^{1/2} R_{s,t}^{1/2}) \in O(m^{1/2} \, \delta^{-1/2})$ & QW \\
\hline
Jarret et al \cite{jarret2018quantum} & $O(R_{s,t}^{1/2} \, C_{s,t}^{1/2})$ & adjacency \\
\hline
\hline
folklore QW sampling & $O(m^{1/2} \, \delta^{-1/2})$ & QW \\
\hline
this work & $\tO(m^{1/3} \, \delta^{-1/3})$ & QW \\
\hline
this work & $\tO(m^{1/3} \, \delta^{-1/3} \, d_M^{1/3})$ & array \\
\hline
\end{tabular}
\caption{Query complexity of $st$-connectivity using different quantum algorithms in different models.
The array model measures the number of degree and neighbor queries; the adjacency model measures the number of pair queries (e.g., ``are $i$ and $j$ neighbors?''); the QW model measures the number of QW steps.
The quantities $d_{s,t}$ and $R_{s,t}$ denote the length of the shortest path and the effective resistance, respectively, between $s$ and $t$.
The quantity $C_{s,t}$ denotes the capacitance between $s$ and $t$ in negative instances, i.e., if $s$ and $t$ are disconnected then $C_{s,t}$ quantifies ``how'' disconnected they are.} \label{fig:st}
\end{figure}

\subsection{Graph Isomorphism} \label{sec:graph-iso}
We consider some given $n$-node graph $g$, described by its adjacency matrix.
To this graph we can associate a new regular graph $G^{(g)} = (\V,\E)$ with nodes $\V = \{\sigma(g) \mid \sigma \in S_n\}$, consisting of permutations of the original graph nodes, and edges $\E = \{(h,\sigma_{i,j}(h)) \mid h \in \V, i,j \in [n]\}$, corresponding to all possible transpositions of two elements.
We can easily prove the following.
\begin{lemma}
The random walk on $G^{(g)}$ has a spectral gap $\delta \in \Omega(n^{-1} \log^{-1} n)$.
\end{lemma}
\begin{proof}
If $|\V| = n!$ (i.e., $g \neq \sigma(g)$ if $\sigma \neq 1$), this graph is isomorphic to the Cayley graph derived from the symmetric group with generators given by transpositions.
The mixing time of a random walk on this graph is $O(n \log n)$ by a result of Diaconis and Shashahani \cite{diaconis1981generating}, implying a lower bound on its spectral gap $\delta \in \Omega(n^{-1} \log^{-1} n)$.

If $|\V| < n!$, the graph is effectively an edge contraction of the random transposition graph.
Following Aldous and Fill \cite[Proposition 4.44]{aldous2002reversible}, a random walk on this graph is an \textit{induced chain} of the random walk on the symmetric group, in particular having a spectral gap lower bounded by the spectral gap of the original walk.
\end{proof}

Next we show how to implement a QW step on $G^{(g)}$ in $\tO(1)$ steps.
By Theorem \ref{thm:qsampling} we can then create a superposition over the edges of $G^{(g)}$ (or, equivalently, its nodes) in time $\tO(m^{1/3}) = \tO(2^{n/3})$, and by Proposition \ref{prop:st} we can solve $st$-connectivity (i.e.~graph isomorphism) in the same time.

\begin{lemma}
Implementing a quantum walk on $G^{(g)}$ takes time $\tO(1)$.
\end{lemma}
\begin{proof}
Since we may have multi-edges, corresponding to permutations that leave the input graph invariant, we will slightly alter the QW to take place on a node+coin space (as in e.g.~\cite{aharonov2001quantum,ambainis2007quantum}) rather than on the edge space.
The relevant spectral properties from Lemma \ref{lem:QW-gap} however remain unchanged, as is easily seen by following for instance the proof of \cite{krovi2016quantum}.
We define the QW node+coin space, associated to the input graph $g$, as $\mathrm{span}_\complex\{\ket{\sigma(g),i,j} \mid \sigma\in S_n, i,j\in[n]\}$, with $S_n$ the symmetric group of permutations.
Similarly to Section \ref{sec:walks}, the QW operator $W = S R_\E$ consists of a reflection $R_\E$ around a subspace $\mathrm{span}_\complex\{\ket{\psi_{\sigma(g)}} \mid \sigma \in S_n\}$, now defined as
\[
\ket{\psi_{\sigma(g)}}
= \frac{1}{n} \sum_{i,j \in [n]} \ket{g,i,j},
\]
and the shift operator $S$ defined by $S\ket{g',i,j} = \ket{\sigma_{i,j}(g'),i,j}$.
Each of these operators can be implemented in $\tO(1)$ steps.
\end{proof}

\section{Acknowledgements}

This work greatly benefited from discussions with Alain Sarlette, Stacey Jeffery, Anthony Leverrier, Ronald de Wolf, Andr\'e Chailloux and Fr\'ed\'eric Magniez.
Part of this work was supported by the CWI-Inria International Lab.

\bibliographystyle{alpha}
\bibliography{/home/simon/Dropbox/biblio}

\newcommand{\etalchar}[1]{$^{#1}$}
\newcommand{\lName}{1}\newcommand{\arxiv}[1]{arXiv:
  \small\href{https://arxiv.org/abs/#1}{\ttfamily{#1}}\?}\def\?#1{\if.#1{}\else#1\fi}\newcommand{\doi}[1]{doi:
  \small\href{https://doi.org/#1}{\ttfamily{\nolinkurl{#1}}}\?}\newcommand{\skp}[3]{#2}\newcommand{\focs
  }[1]{\if\lName1\skp{ }{Proceedings of the #1 {IEEE} Symposium on Foundations
  of Computer Science ({FOCS})}{ }\else{FOCS}\fi}\newcommand{\stoc
  }[1]{\if\lName1\skp{ }{Proceedings of the #1 {ACM} Symposium on Theory of
  Computing ({STOC})}{ }\else{STOC}\fi}\newcommand{\soda }[1]{\if\lName1\skp{
  }{Proceedings of the #1 {ACM-SIAM} Symposium on Discrete Algorithms
  ({SODA})}{ }\else{SODA}\fi}\newcommand{\stacs }[1]{\if\lName1\skp{
  }{Proceedings of the #1 Symposium on Theoretical Aspects of Computer Science
  ({STACS})}{ }\else{STACS}\fi}\newcommand{\itcs }[1]{\if\lName1\skp{
  }{Proceedings of the #1 Innovations in Theoretical Computer Science
  Conference (ITCS)}{ }\else{ITCS}\fi}\newcommand{\fsttcs }[1]{\if\lName1\skp{
  }{Proceedings of the #1 International Conference on Foundations of Software
  Technology and Theoretical Computer Science (FSTTCS)}{
  }\else{FSTTCS}\fi}\newcommand{\ccc }[1]{\if\lName1\skp{ }{Proceedings of the
  #1 {IEEE} Conference on Computational Complexity ({CCC})}{
  }\else{CCC}\fi}\newcommand{\colt }[1]{\if\lName1\skp{ }{Proceedings of the #1
  Conference On Learning Theory (COLT)}{ }\else{COLT}\fi}\newcommand{\nips
  }[1]{\if\lName1\skp{ }{Advances in Neural Information Processing Systems #1
  ({NIPS})}{ }\else{NIPS}\fi}\newcommand{\aistats }[1]{\if\lName1\skp{
  }{Proceedings of the #1 International Conference on Artificial Intelligence
  and Statistics ({AISTATS})}{ }\else{AISTATS}\fi}\newcommand{\icalp
  }[1]{\if\lName1\skp{ }{Proceedings of the #1 International Colloquium on
  Automata, Languages, and Programming (ICALP)}{
  }\else{ICALP}\fi}\newcommand{\icml }[1]{\if\lName1\skp{ }{Proceedings of the
  #1 International Conference on Machine Learning (ICML)}{
  }\else{ICML}\fi}\newcommand{\esa }[1]{\if\lName1\skp{ }{Proceedings of the #1
  European Symposium on Algorithms (ESA)}{ }\else{ESA}\fi}\newcommand{\jacm
  }{\if\lName1\skp{ }{Journal of the ACM}{ }\else{J. ACM}\fi}\newcommand{\jams
  }{\if\lName1\skp{ }{Journal of the AMS}{ }\else{J. AMS}\fi}\newcommand{\pams
  }{\if\lName1\skp{ }{Proceedings of the AMS}{ }\else{Proc.
  AMS}\fi}\newcommand{\linalgappl }{\if\lName1\skp{ }{Linear Algebra and its
  Applications}{ }\else{Lin. Alg. \& App.}\fi}\newcommand{\jalgo
  }{\if\lName1\skp{ }{Journal of Algorithms}{ }\else{J.
  Alg.}\fi}\newcommand{\jcss }{\if\lName1\skp{ }{Journal of Computer and System
  Sciences}{}\else{J. Comp. Sys. Sci.}\fi}\newcommand{\cc }{\if\lName1\skp{
  }{Computational Complexity}{ }\else{Comp. Comp.}\fi}\newcommand{\algor
  }{\if\lName1\skp{ }{Algorithmica}{ }\else{Alg.}\fi}\newcommand{\comba
  }{\if\lName1\skp{ }{Combinatorica}{ }\else{Comb.}\fi}\newcommand{\cacm
  }{\if\lName1\skp{ }{Communications of the ACM}{ }\else{Comm.
  ACM}\fi}\newcommand{\sigart }{\if\lName1\skp{ }{SIGART Bulletin}{
  }\else{SIGART Bull.}\fi}\newcommand{\sigactn }{\if\lName1\skp{ }{SIGACT
  News}{ }\else{SIGACT News}\fi}\newcommand{\eatcsbul }{\if\lName1\skp{
  }{Bulletin of the {EATCS}}{ }\else{Bull. {EATCS}}\fi}\newcommand{\siamrev
  }{\if\lName1\skp{ }{SIAM Review}{ }\else{SIAM Rev.}\fi}\newcommand{\siamjc
  }{\if\lName1\skp{ }{SIAM Journal on Computing}{ }\else{SIAM J.
  Comp.}\fi}\newcommand{\siamjo }{\if\lName1\skp{ }{SIAM Journal on
  Optimization}{ }\else{SIAM J. Opt.}\fi}\newcommand{\siamjdm }{\if\lName1\skp{
  }{SIAM Journal on Discrete Mathematics}{ }\else{SIAM J. Disc.
  Math.}\fi}\newcommand{\siamjnum }{\if\lName1\skp{ }{SIAM Journal on Numerical
  Analysis}{ }\else{SIAM J. Num. Anal.}\fi}\newcommand{\siamjmathanal
  }{\if\lName1\skp{ }{SIAM Journal on Mathematical Analysis}{ }\else{SIAM J.
  Math. Anal.}\fi}\newcommand{\discmath }{\if\lName1\skp{ }{Discrete
  Mathematics}{ }\else{Disc. Math.}\fi}\newcommand{\das }{\if\lName1\skp{
  }{Discrete Applied Mathematics}{ }\else{Disc. App.
  Math.}\fi}\newcommand{\amatstat }{\if\lName1\skp{ }{Annals of Mathematical
  Statistics}{ }\else{Ann. Math. Stat.}\fi}\newcommand{\rms }{\if\lName1\skp{
  }{Russian Mathematical Surveys}{ }\else{Russ. Math.
  Surv.}\fi}\newcommand{\invmath }{\if\lName1\skp{ }{Inventiones Mathematicae}{
  }\else{Inv. Math.}\fi}\newcommand{\jnumber }{\if\lName1\skp{ }{Journal of
  Number Theory}{ }\else{J. Num. Th.}\fi}\newcommand{\toc }{\if\lName1\skp{
  }{Theory of Computing}{ }\else{Th. Comp.}\fi}\newcommand{\quantum
  }{\if\lName1\skp{ }{{Quantum}}{ }\else{Quant.}\fi}\newcommand{\cmp
  }{\if\lName1\skp{ }{Communications in Mathematical Physics}{ }\else{Comm.
  Math. Phys.}\fi}\newcommand{\rspa }{\if\lName1\skp{ }{Proceedings of the
  Royal Society A}{ }\else{Proc. Roy. Soc. A}\fi}\newcommand{\qic
  }{\if\lName1\skp{ }{Quantum Information and Computation}{ }\else{Quant. Inf.
  \& Comp.}\fi}\newcommand{\physrev }{\if\lName1\skp{ }{Physical Review}{
  }\else{Phys. Rev.}\fi}\newcommand{\pra }{\if\lName1\skp{ }{Physical Review
  A}{ }\else{Phys. Rev. A}\fi}\newcommand{\prb }{\if\lName1\skp{ }{Physical
  Review B}{ }\else{Phys. Rev. B}\fi}\newcommand{\pre }{\if\lName1\skp{
  }{Physical Review E}{ }\else{Phys. Rev. E}\fi}\newcommand{\prx
  }{\if\lName1\skp{ }{Physical Review X}{ }\else{Phys. Rev.
  X}\fi}\newcommand{\prl }{\if\lName1\skp{ }{Physical Review Letters}{
  }\else{Phys. Rev. Lett.}\fi}\newcommand{\physrep }{\if\lName1\skp{ }{Physics
  Reports}{ }\else{Phys. Rep.}\fi}\newcommand{\rmp }{\if\lName1\skp{ }{Reviews
  of Modern Physics}{ }\else{Rev. Mod. Phys. }\fi}\newcommand{\phystoday
  }{\if\lName1\skp{ }{Physics Today}{ }\else{Phys.
  Today}\fi}\newcommand{\physics }{\if\lName1\skp{ }{Physics}{
  }\else{Phys.}\fi}\newcommand{\nature }{\if\lName1\skp{ }{Nature}{
  }\else{Nat.}\fi}\newcommand{\natcomm }{\if\lName1\skp{ }{Nature
  Communications}{ }\else{Nat. Comm.}\fi}\newcommand{\natphys }{\if\lName1\skp{
  }{Nature Physics}{ }\else{Nat. Phys.}\fi}\newcommand{\npjqi }{\if\lName1\skp{
  }{npj Quantum Information}{ }\else{npj Quant. Inf.}\fi}\newcommand{\scirep
  }{\if\lName1\skp{ }{Scientific Reports}{ }\else{Sci.
  Rep.}\fi}\newcommand{\science }{\if\lName1\skp{ }{Science}{
  }\else{Sci.}\fi}\newcommand{\jpa }{\if\lName1\skp{ }{Journal of Physics A:
  Mathematical and Theoretical}{ }\else{J. Phys. A}\fi}\newcommand{\ijtp
  }{\if\lName1\skp{ }{International Journal of Theoretical Physics}{
  }\else{Int. J. Th. Phys.}\fi}\newcommand{\jmo }{\if\lName1\skp{ }{Journal of
  Modern Optics}{ }\else{J. Mod. Opt.}\fi}\newcommand{\jstatph
  }{\if\lName1\skp{ }{Journal of Statistical Physics}{ }\else{J. Stat.
  Phys.}\fi}\newcommand{\lncs }{\if\lName1\skp{ }{Lecture Notes in Computer
  Science}{ }\else{L. Notes Comp. Sci.}\fi}\newcommand{\lnai }{\if\lName1\skp{
  }{Lecture Notes in Artificial Intelligence}{ }\else{L. Notes Art.
  Int.}\fi}\newcommand{\lnm }{\if\lName1\skp{ }{Lecture Notes in Mathematics}{
  }\else{L. Notes Math.}\fi}\newcommand{\tams }{\if\lName1\skp{ }{Transactions
  of the American Mathematical Society}{ }\else{Trans.
  AMS}\fi}\newcommand{\ieeeit }{\if\lName1\skp{ }{{IEEE} Transactions on
  Information Theory}{ }\else{{IEEE} Trans. Inf. Th.}\fi}\newcommand{\iscs
  }{\if\lName1\skp{ }{International Series in Computer Science}{ }\else{Int.
  Ser. Comp. Sci.}\fi}\newcommand{\tocl }{\if\lName1\skp{ }{Theory of Computing
  Library}{ }\else{Th. Comp. Lib.}\fi}
\begin{thebibliography}{VAGGdW17}

\bibitem[AAKV01]{aharonov2001quantum}
Dorit Aharonov, Andris Ambainis, Julia Kempe, and Umesh Vazirani.
\newblock Quantum walks on graphs.
\newblock In {\em \stoc{33rd}}, pages 50--59. ACM, 2001.
\newblock \arxiv{quant-ph/0012090}.

\bibitem[ACL11]{ambainis2011quantum}
Andris Ambainis, Andrew~M Childs, and Yi-Kai Liu.
\newblock Quantum property testing for bounded-degree graphs.
\newblock In {\em Approximation, Randomization, and Combinatorial Optimization.
  Algorithms and Techniques}, pages 365--376. Springer, 2011.
\newblock \arxiv{1012.3174}.

\bibitem[AF02]{aldous2002reversible}
David Aldous and Jim Fill.
\newblock Reversible {M}arkov chains and random walks on graphs.
\newblock Unfinished monograph, 2002.
\newblock
  \small\href{https://www.stat.berkeley.edu/~aldous/RWG/book.pdf}{\ttfamily{link}}.

\bibitem[Amb07]{ambainis2007quantum}
Andris Ambainis.
\newblock Quantum walk algorithm for element distinctness.
\newblock {\em \siamjc}, 37(1):210--239, 2007.
\newblock \arxiv{quant-ph/0311001}.

\bibitem[AMRR11]{ambainis2011symmetry}
Andris Ambainis, Lo{\"\i}ck Magnin, Martin Roetteler, and J{\'e}r{\'e}mie
  Roland.
\newblock Symmetry-assisted adversaries for quantum state generation.
\newblock In {\em \ccc{26th}}, pages 167--177. IEEE, 2011.
\newblock \arxiv{1012.2112}.

\bibitem[AP09]{andersen2009finding}
Reid Andersen and Yuval Peres.
\newblock Finding sparse cuts locally using evolving sets.
\newblock In {\em \stoc{41st}}, pages 235--244. ACM, 2009.
\newblock \arxiv{0811.3779}.

\bibitem[AS19]{apers2018quantum}
Simon Apers and Alain Sarlette.
\newblock Quantum fast-forwarding {Markov} chains and property testing.
\newblock {\em \qic}, 19(3\&4):181--213, 2019.
\newblock \arxiv{1804.02321}.

\bibitem[ATS03]{aharonov2003adiabatic}
Dorit Aharonov and Amnon Ta-Shma.
\newblock Adiabatic quantum state generation and statistical zero knowledge.
\newblock In {\em \stoc{35th}}, pages 20--29. ACM, 2003.
\newblock \arxiv{quant-ph/0301023}.

\bibitem[Bab91]{babai1991local}
L{\'a}szl{\'o} Babai.
\newblock Local expansion of vertex-transitive graphs and random generation in
  finite groups.
\newblock In {\em \stoc{23rd}}, volume~91, pages 164--174. ACM, 1991.
\newblock \doi{10.1145/103418.103440}.

\bibitem[Bab16]{babai2016graph}
L{\'a}szl{\'o} Babai.
\newblock Graph isomorphism in quasipolynomial time.
\newblock In {\em \stoc{48th}}, pages 684--697. ACM, 2016.
\newblock \arxiv{1512.03547}.

\bibitem[Bel13]{belovs2013quantum}
Aleksandrs Belovs.
\newblock Quantum walks and electric networks.
\newblock \arxiv{1302.3143}, 2013.

\bibitem[BFR{\etalchar{+}}13]{batu2013testing}
Tu{\u{g}}kan Batu, Lance Fortnow, Ronitt Rubinfeld, Warren~D Smith, and Patrick
  White.
\newblock Testing closeness of discrete distributions.
\newblock {\em \jacm}, 60(1):4, 2013.
\newblock \arxiv{1009.5397}.

\bibitem[BHMT02]{brassard2002quantum}
Gilles Brassard, Peter H{\o}yer, Michele Mosca, and Alain Tapp.
\newblock Quantum amplitude amplification and estimation.
\newblock {\em Contemporary Mathematics}, 305:53--74, 2002.
\newblock \arxiv{quant-ph/0005055}.

\bibitem[BHT97]{brassard1997quantum}
Gilles Brassard, Peter H{\o}yer, and Alain Tapp.
\newblock Quantum algorithm for the collision problem.
\newblock {\em ACM SIGACT News (Cryptology Column)}, 28:14--19, 1997.
\newblock \arxiv{quant-ph/9705002}.

\bibitem[BR12]{belovs2012span}
Aleksandrs Belovs and Ben~W Reichardt.
\newblock Span programs and quantum algorithms for st-connectivity and claw
  detection.
\newblock In {\em \esa{20th}}, pages 193--204. Springer, 2012.
\newblock \arxiv{1203.2603}.

\bibitem[CDK{\etalchar{+}}16]{chiericetti2016sampling}
Flavio Chiericetti, Anirban Dasgupta, Ravi Kumar, Silvio Lattanzi, and
  Tam{\'a}s Sarl{\'o}s.
\newblock On sampling nodes in a network.
\newblock In {\em Proceedings of the 25th International Conference on World
  Wide Web (WWW)}, pages 471--481. International WWW Conferences, 2016.
\newblock \doi{10.1145/2872427.2883045}.

\bibitem[CMB16]{cade2016time}
Chris Cade, Ashley Montanaro, and Aleksandrs Belovs.
\newblock Time and space efficient quantum algorithms for detecting cycles and
  testing bipartiteness.
\newblock \arxiv{1610.00581}, 2016.

\bibitem[CPS15]{czumaj2015testing}
Artur Czumaj, Pan Peng, and Christian Sohler.
\newblock Testing cluster structure of graphs.
\newblock In {\em \stoc{47th}}, pages 723--732. ACM, 2015.
\newblock \arxiv{1504.03294}.

\bibitem[DHHM06]{durr2006quantum}
Christoph D{\"u}rr, Mark Heiligman, Peter H{\o}yer, and Mehdi Mhalla.
\newblock Quantum query complexity of some graph problems.
\newblock {\em \siamjc}, 35(6):1310--1328, 2006.
\newblock \arxiv{quant-ph/0401091}.

\bibitem[DS81]{diaconis1981generating}
Persi Diaconis and Mehrdad Shahshahani.
\newblock Generating a random permutation with random transpositions.
\newblock {\em Probability Theory and Related Fields}, 57(2):159--179, 1981.
\newblock \doi{10.1007/BF00535487}.

\bibitem[GR02]{goldreich1997property}
Oded Goldreich and Dana Ron.
\newblock Property testing in bounded degree graphs.
\newblock {\em \algor}, 32(2):302--343, 2002.
\newblock \doi{10.1007/s00453-001-0078-7}.

\bibitem[GR11]{goldreich2011testing}
Oded Goldreich and Dana Ron.
\newblock On testing expansion in bounded-degree graphs.
\newblock In {\em Studies in Complexity and Cryptography. Miscellanea on the
  Interplay between Randomness and Computation}, pages 68--75. Springer, 2011.
\newblock \doi{10.1007/978-3-642-22670-0_9}.

\bibitem[JJKP18]{jarret2018quantum}
Michael Jarret, Stacey Jeffery, Shelby Kimmel, and Alvaro Piedrafita.
\newblock Quantum algorithms for connectivity and related problems.
\newblock In {\em \esa{26th}}, pages 49:1--49:13. Springer, 2018.
\newblock \arxiv{1804.10591}.

\bibitem[KMOR16]{krovi2016quantum}
Hari Krovi, Fr{\'e}d{\'e}ric Magniez, Maris Ozols, and J{\'e}r{\'e}mie Roland.
\newblock Quantum walks can find a marked element on any graph.
\newblock {\em \algor}, 74(2):851--907, 2016.
\newblock \arxiv{1002.2419}.

\bibitem[KP16]{kerenidis2016quantum}
Iordanis Kerenidis and Anupam Prakash.
\newblock Quantum recommendation systems.
\newblock In {\em \itcs{8th}}, pages 49:1--49:21. Schloss
  Dagstuhl--Leibniz-Zentrum fuer Informatik, 2016.
\newblock \arxiv{1603.08675}.

\bibitem[LPW17]{levin2017markov}
David~A Levin, Yuval Peres, and Elizabeth~L Wilmer.
\newblock {\em {M}arkov chains and mixing times}.
\newblock American Mathematical Society, 2017.
\newblock \doi{10.1090/mbk/058}.

\bibitem[Lut11]{lutomirski2011component}
Andrew Lutomirski.
\newblock Component mixers and a hardness result for counterfeiting quantum
  money.
\newblock \arxiv{1107.0321}, 2011.

\bibitem[MNRS11]{magniez2011search}
Fr{\'e}d{\'e}ric Magniez, Ashwin Nayak, J{\'e}r{\'e}mie Roland, and Miklos
  Santha.
\newblock Search via quantum walk.
\newblock {\em \siamjc}, 40(1):142--164, 2011.
\newblock \arxiv{quant-ph/0608026}.

\bibitem[OBD18]{orsucci2015faster}
Davide Orsucci, Hans~J. Briegel, and Vedran Dunjko.
\newblock Faster quantum mixing for slowly evolving sequences of {M}arkov
  chains.
\newblock {\em \quantum}, 2:105, 2018.
\newblock \arxiv{1503.01334}.

\bibitem[PW09]{poulin2009sampling}
David Poulin and Pawel Wocjan.
\newblock Sampling from the thermal quantum gibbs state and evaluating
  partition functions with a quantum computer.
\newblock {\em \prl}, 103(22):220502, 2009.
\newblock \arxiv{0905.2199}.

\bibitem[Ric07]{richter2007quantum}
Peter~C Richter.
\newblock Quantum speedup of classical mixing processes.
\newblock {\em \pra}, 76(4):042306, 2007.
\newblock \arxiv{quant-ph/0609204}.

\bibitem[SBBK08]{somma2008quantum}
Rolando~D Somma, Sergio Boixo, Howard Barnum, and Emanuel Knill.
\newblock Quantum simulations of classical annealing processes.
\newblock {\em Physical Review Letters}, 101(13):130504, 2008.
\newblock \arxiv{0804.1571}.

\bibitem[Sin12]{sinclair2012algorithms}
Alistair Sinclair.
\newblock {\em Algorithms for random generation and counting: a {M}arkov chain
  approach}.
\newblock Springer Science \& Business Media, 2012.
\newblock \doi{10.1007/978-1-4612-0323-0}.

\bibitem[ST13]{spielman2013local}
Daniel~A Spielman and Shang-Hua Teng.
\newblock A local clustering algorithm for massive graphs and its application
  to nearly linear time graph partitioning.
\newblock {\em \siamjc}, 42(1):1--26, 2013.
\newblock \arxiv{0809.3232}.

\bibitem[Sze04]{szegedy2004quantum}
Mario Szegedy.
\newblock Quantum speed-up of {M}arkov chain based algorithms.
\newblock In {\em \focs{45th}}, pages 32--41. IEEE, 2004.
\newblock \arxiv{quant-ph/0401053}.

\bibitem[VAGGdW17]{van2017quantum}
Joran Van~Apeldoorn, Andr{\'a}s Gily{\'e}n, Sander Gribling, and Ronald
  de~Wolf.
\newblock Quantum {SDP}-solvers: better upper and lower bounds.
\newblock In {\em \focs{58th}}, pages 403--414. IEEE, 2017.
\newblock \arxiv{1705.01843}.

\bibitem[WA08]{wocjan2008speedup}
Pawel Wocjan and Anura Abeyesinghe.
\newblock Speedup via quantum sampling.
\newblock {\em \pra}, 78(4):042336, 2008.
\newblock \arxiv{0804.4259}.

\bibitem[Wat00]{watrous2000succinct}
John Watrous.
\newblock Succinct quantum proofs for properties of finite groups.
\newblock In {\em \focs{41st}}, pages 537--546. IEEE, 2000.
\newblock \arxiv{cs/0009002}.

\bibitem[Wat01]{watrous2001quantum}
John Watrous.
\newblock Quantum simulations of classical random walks and undirected graph
  connectivity.
\newblock {\em \jcss}, 62(2):376--391, 2001.
\newblock \arxiv{cs/9812012}.

\end{thebibliography}

\end{document}